\documentclass{paper}
\title{An Incentive-Compatible Smart Contract for Decentralized Commerce}
\author {
	Nikolaj I. Schwartzbach\thanks{Supported by the ERC Advanced Grant MPCPRO.} \\
	\small Department of Computer Science, Aarhus University
}

\usepackage{amsmath, amssymb, amsthm}
\newtheorem{definition}{Definition}
\newtheorem{lemma}{Lemma}
\newtheorem{theorem}{Theorem}

\usepackage[capitalise]{cleveref}

\usepackage{tikz}
\usetikzlibrary{arrows.meta}
\usetikzlibrary{calc,shapes.callouts,shapes.arrows}
\newcommand{\gear}[5]{%
	\foreach \i in {1,...,#1} {%
		[rotate=(\i-1)*360/#1]  (0:#2)  arc (0:#4:#2) {[rounded corners=1.5pt]
			-- (#4+#5:#3)  arc (#4+#5:360/#1-#5:#3)} --  (360/#1:#2)
}}  
\def\foldedpaper#1{
	\tikz[scale=#1,line width={#1*1pt}]{
		\def\a{1.41} % relative height
		\def\b{0.2}  % relative height/width of corner
		\def\c{0.1}  % relative margin width (on either side)
		\def\d{0.05} % vertical offset of lines
		\def\N{6}    % number of lines
		\draw         (0,0)
		--  ++(-1,0)
		--  ++(0,\a)
		--  ++(1-\b,0)
		--  ++(\b,-\b)
		-- cycle;
		\foreach \lnum in {1,...,\N}{
			\pgfmathsetmacro\yline{\a-\d-\lnum*\a/(\N+1)}
			\draw (-1+\c,\yline) -- (-\c,\yline);
		}
		\draw[fill=white] (0,\a-\b) -- ++(-\b,0) -- ++ (0,\b);
	}
}
\newcommand{\bubbleright}[2]{
	\tikz[remember picture,baseline]{\node[anchor=base,inner sep=0,outer sep=0]%
		(a) {#1};\node[overlay,cloud callout,callout relative pointer={(-0.2cm,-0.2cm)},%
		aspect=1.5,draw] at ($(a.north)+(0.8cm,0.4cm)$) {#2};}%
}%

\newcommand{\bubbleleft}[2]{
	\tikz[remember picture,baseline]{\node[anchor=base,inner sep=0,outer sep=0]%
		(a) {#1};\node[overlay,cloud callout,callout relative pointer={(0.2cm,-0.2cm)},%
		aspect=1.5,draw] at ($(a.north)-(0.8cm,-0.4cm)$) {#2};}%
}%

\renewcommand{\P}{\textsf{P}}
\providecommand{\keywords}[1]
{
	{
	\textbf{\textit{Keywords---}} #1}
}

\begin{document}

	\maketitle
	
	\begin{abstract}
		We propose a smart contract that allows two mutually distrusting parties to transact any non-digital good or service by deploying a smart contract on a blockchain to act as escrow. The contract settles disputes by letting parties wager that they can convince an arbiter that they were the honest party. We analyse the contract as an extensive-form game and prove that the honest strategy is secure in a strong game-theoretic sense if and only if the arbiter is biased in favor of honest parties. By relaxing the security notion, we can replace the arbiter by a random coin toss. Finally, we show how to generalize the contract to multiparty transactions in a way that amortizes the transaction fees.
	\end{abstract}

\vspace{-0.5cm}
	
	\section{Introduction}
	A fundamental problem of electronic commerce is ensuring both ends of the trade are upheld: an honest seller should always receive payment, and an honest buyer should only pay if the seller was honest. Traditionally, this is ensured by introducing a trusted intermediary who holds the payment in escrow until the trade has completed, after which the funds are released to the seller. The parties are typically required not to be anonymous, so as to enable either party to hold the other party accountable in case of fraudulent behavior, and potentially subject to legal repercussions. This, in conjunction with reputation systems, has proved to be an effective means to ensure honest and efficient trading, as evidenced by the enormous market cap of online marketplaces such as Amazon or Alibaba. However, this crucially relies on being able to trust the intermediary to behave honestly: while the intermediary has strong incentive to maintain a good reputation, from a cryptographic point of view this does not address the fundamental issue. A central marketplace still has incentive to engage in monopolistic behavior, such as removal of competitors' products, or differential pricing based on customer demographics, to the extent that it remains undetected. 
	
	Recent years has seen the creation of \emph{darknet markets} that take advantage of cryptocurrency and mix networks to provide decentralized and somewhat anonymous trade of goods and services. They arguably solve some issues with central marketplaces, but in doing so, also enable black market/criminal activity to remain relatively unchecked. The most infamous darknet market was ``Silk Road'', known for selling illegal goods such as drugs and weapons. It operated from February 2011, until it was seized by authorities in October 2013, and the developer, Ross Ulbricht, was sentencted to double life imprisonment. But this is a rarity: due to the anonymous nature of the markets, it is often difficult to prosecute individuals, and many convictions of buyers are based on some sort of circumstantial metadata, such as recently having used their credit card to purchase cryptocurrency for a similar amount. However, most darknet markets remain inherently centralized, in that all data as well as escrowed funds are processed directly by the market itself, essentially at its mercy. Buyers and sellers are required to trust both the benevolence and competence of the market, a trust which is at best misplaced and at worst disastrous in consequence. Indeed, there are numerous examples of prominent darknet markets being hacked and all funds held in escrow stolen, or the operators of the market themselves perform an \emph{exit scam}, i.e. suddenly stealing the funds in escrow and subsequently closing the market. It is often difficult, if not impossible, to recover the stolen funds and hold anyone accountable.
	
	In this paper, we consider a seller who wants to sell an item to a buyer without having to rely on a trusted third party. We assume both parties have access to a blockchain that allows them to deploy smart contracts that can exchange cryptocurrency. Our goal is to replace the trusted third party with a smart contract, such that parties can be trusted to complete their end of the trade. Specifically, we want to design a smart contract for which we can prove that parties have an economic incentive to behave honestly. Or in other words, 
	\emph{can we design a smart contract to facilitate decentralized trading of non-digital goods and services in a way that provably ensures honest behavior in rational agents?}
	
	\subsection{Our results}
	We propose a smart contract for escrow of funds that enables any two parties to engage in the trade of a physical good or service in exchange for cryptocurrency. The contract relies crucially on an \emph{arbiter} that is invoked only in the case of a dispute. The purpose of the arbiter is to distinguish the honest party from the dishonest party. The basic idea is that either party can issue a dispute by making a ``wager'' of size $\lambda$ that they will win the arbitration: the winner is repaid their deposit as well as the funds held in escrow. We prove that both buyer and seller are incentivized to behave honestly if and only if the arbiter is not too biased against honest parties. Specifically, let $\gamma$ be the ``error rate'' of the arbiter: then we show there is a value of $\lambda$ such that the contract has strong game-theoretic security if and only if $\gamma < \frac{1}{2}$. By instead considering a weaker notion of security, we can use a random coin flip as arbitration. We sketch a simple construction based on Blum's coin toss protocol.
	
	The contract can be run on any blockchain that supports smart contracts (such as Ethereum). As a result, many properties (anonymity, efficiency etc.) of our smart contract are inherited by the corresponding blockchain. We feature a discussion of different ways to instantiate the smart contract. In particular, the contract can be used in a manner that complies with current laws and regulations by using a blockchain with revocable anonymity: a party who participates in distributing illicit goods can be deanonymized by the courts, while all other parties remain anonymous. 
	
	\subsection{Related work} 
	A variety of solutions have been proposed for replacing the trusted third party by a smart contract in so-called \emph{atomic swaps}. Most academic work has focused on \emph{digital goods}, the delivery of which can be deterministically determined. 
	
	Dziembowski, Eckey and Faust propose a protocol, called \emph{FairSwap} \cite{fairswap}, with essentially optimal security: the goods are delivered to the buyer if and only if the seller receives the money. Their solution relies heavily on cryptography and assumes the goods can be represented as a finite field element. As a result, their protocol does not generalize in any meaningful way to physical goods. It seems unlikely we can achieve this notion of security for non-digital goods, due to a fundmantal discrepancy between the physical and the digital world.
	
	Asgaonkar and Krishnamachari propose a smart contract for the trade of digital goods \cite{buyer_seller_dilemma}: both parties make a deposit of funds \emph{a priori} (a \emph{dual-deposit}) which is only refunded if the trade was successful. They prove that the honest strategy is the unique subgame perfect equilibrium for sufficiently large deposits. Like FairSwap, their solution only works for digital goods. 
	
	Witkowski, Seuken and Parkes consider the setting of escrow in online auctions \cite{incentive_compatible_escrow}. Their idea is to pay some of the buyers a rebate to offset their expected loss from engaging in a transaction with the seller. Whether a buyer is paid a rebate depends on the reports of other buyers. They prove that the seller has strict incentive to be honest, while the buyers are only weakly incentivized to do so. They show that strict incentives for the buyers is possible if the escrow has distributional knowledge about the variations in seller abilities, based on a \emph{peer prediction} method. Unfortunately, their solutions rely on a somewhat idealized setting in which there are many buyers concurrently transacting with the same seller, as otherwise buyers and/or sellers may have an incentive to collude, thus breaking security. In addition, it is not obvious how to apply their work to a non-auction type setting.
	
	Outside academic circles, there are several proposed solutions, of which the most promising is OpenBazaar. It is also blockchain based and as such provides some level of decentralization. However, its dispute resolution remains centralized in a sense, since all moderation is done by a set of trusted moderators, requiring buyers and sellers to blindly trust the benevolence and competence of these moderators. In addition, it does not have any formal analysis of correctness or security, and thus falls short in rigorously solving the buyer and seller's dilemma. To the best knowledge of the author, there is no ``truly decentralized'' market at the time of writing.
	
	\section{The basic contract}
	
	In this section we describe our contract for trade of non-digital goods and services. We consider a buyer $B$ who wants to purchase an item $\mathit{it}$ from a seller $S$. The item can be a physical good (a book), or a service (roof repair). The item is sold for a price of $x$, and has a ``perceived value'' to the buyer of $y > x$, while the seller perceives the value at $x' < x$\footnote{From a game-theoretic point of view, we have to assume $y>x>x'$, as otherwise neither buyer nor seller has no incentive to engage in the transaction.} The item $\mathit{it}$ is \emph{non-digital} which means it has to be shipped through a physical channel ``off-chain''. See \cref{fig:contract} for an illustration. Unfortunately, there is no computer program that can rigorously determine whether or not $\mathit{it}$ was physically delivered to the buyer because of a fundamental gap between the digital and the physical world. We assume both parties have access to a blockchain, which for our purposes is a shared data structure that allows both parties to deploy a smart contract $\pi$, that can maintain state, respond to queries, and transfer funds. Unlike a human third party, the smart contract can be guaranteed to behave honestly due to security of the underlying blockchain. For simplicity, we assume the blockchain is secure and incorruptible, and consider only attacks on the contract itself. In addition, we assume transaction fees are negligible compared to the items being transacted, such that they can be disregarded entirely.
	
	\begin{figure}
		\centering
		\begin{tikzpicture}
		\node[label={below:{$\mathit{it}$}}] (it) at (-2.8,0) {$\foldedpaper{.3}$};
		\node[draw,circle] (S) at (-2,0) {$S$};
		\draw \gear{10}{0.28}{0.36}{8}{2};
		\node {$\pi$};
		\node[draw,circle] (B) at (2,0) {$B$};
		\node[draw,circle] (A) at (0,-1.25) {$\mathcal{A}$};
		\node at (1,0.2) {$x$};
		\node at (-1,0.2) {$x$};
		\node at (0,1.5) {$it$};
		\draw[->] (B) -- (0.36,0);
		\draw[->] (-0.36,0) -- (S);
		\draw[->, dashed] (S) to[bend left=60] (B);
		\draw[dotted] (S) to (A);
		\draw[dotted] (B) to (A);
		\draw[dotted] (0,-.36) to (A);
		\node at (2.25,0.65) {$\bubbleright{~}{$\foldedpaper{.2} = y$}$};
		\node at (-2.3,.675) {$\bubbleleft{~}{$\foldedpaper{.2} = x'$}$};
		\end{tikzpicture}
		\caption{\small A seller $S$ and a buyer $B$ engaged in the transaction of the non-digital good $\mathit{it}$, using a smart contract $\pi$ and arbiter $\mathcal{A}$. The item is sold for $x$ funds and has a perceived value of $y>x$ to $B$. The money $x$ is transferred from $B$ to $S$ through the contract $\pi$. The dashed line is a unidirectional ``off-chain'' channel, through which $\mathit{it}$ can be sent. The dotted lines indicate that $\mathcal{A}$ is only invoked in case of disputes.}
		\label{fig:contract}
	\end{figure}
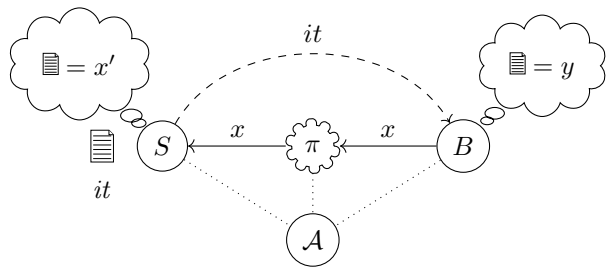
	
	The contract is parameterized by an \emph{arbiter} $\mathcal{A}$, which is a protocol invoked in case of disputes: its purpose is to distinguish the honest party from the dishonest party. We denote by $\gamma$ the \emph{(perceived) upper bound} on the error rate of the arbiter by any of the two parties. This means that each party has reason to think the error rate of the arbiter is $\leq \gamma$, for some constant $0\leq \gamma < 1$ that they provide as input to the contract\footnote{In practice, $B$ and $S$ may have different upper bounds, say $\gamma_B, \gamma_S$. In this case we can let $\gamma = \max(\gamma_B, \gamma_S)$.}. In addition, the contract is parameterized by a \emph{wager constant} $\lambda>0$. In a nutshell the contract proceeds as follows: both parties sign a contract committing to making the trade, and $B$ places $x$ money in escrow. $S$ then delivers $\mathit{it}$ to $B$ ``off-chain'', and notifies the smart contract to transfer the funds in escrow to $S$ and terminate the contract. If $S$ does not deliver $\mathit{it}$ to $B$, then $B$ can trigger a dispute by placing a ``wager'' of size $\lambda$ that they can convince the arbiter that they were the honest party. If $S$ does not respond (or forfeits), it is assumed $\mathit{it}$ was not delivered to $B$, and the contract refunds $x+\lambda$ funds to $B$. However, a dishonest buyer may trigger the dispute phase even when they actually received $\mathit{it}$. In this case, the honest $S$ may counter the wager by also placing a wager of size $\lambda$ that they will win the arbitration. At this point, both parties submit an evidence string and invoke the arbiter. The winner is repaid $x+\lambda$, while the loser receives nothing (the left over $\lambda$ are used to compensate the arbiter for their time). We handle crashing by having timeouts in the contract, in a way that favors the party that did not crash; a buyer that crashes is assumed to have received $\mathit{it}$. Likewise, a seller who fails to respond to a dispute is assumed to forfeit.
	
	\section{Game-theoretic analysis}
	
	We now turn to instantiate the basic smart contract as to achieve security in a game-theoretic sense. Unlike the standard cryptographic model, where parties can be partitioned into honest and dishonest party, instead we assume all parties are \emph{rational}, meaning they seek to maximize their own utility with no concern for the intended behavior of the protocol designer: a rational party will at each point in the execution of the protocol choose the action that maximizes their expected utility. We say the protocol is secure when the maximal utility is achieved when a party behaves honestly. In the following we assume familiarity with extensive form games and refer to \cite{Osborne1994} for details. 
	
	We consider an $n$-party protocol $\pi$ where each party $\P_i$ has a set $\mathcal{S}_i$ of possible strategies, of which there is a unique \emph{honest strategy} $s_i^* \in \mathcal{S}_i$. We define $\mathcal{S} = \mathcal{S}_1 \times \cdots \times \mathcal{S}_n$ as the set of strategy profiles, and let $s^* = (s_1^*, \ldots, s_n^*) \in \mathcal{S}$ be the unique honest strategy profile. We say $\pi$ is secure if $s^*$ is guaranteed to maximize the utility of all parties, robust to unilateral deviations by any single party. Formally, we require that $s^*$ is the (unique) subgame perfect equilibrium: at each subgame, every party receives a strictly smaller utility by deviating from $s^*$. This is likely not sufficient in itself: if the incentive to choose $s^*$ is too small, then there may be other reasons to deviate not captured by the utilities of the game. If the difference is sufficiently large (say $>\varepsilon$) then we say the protocol is secure in game-theoretic sense against $\varepsilon$-deviating rational adversaries. 
	
	\begin{definition}[Strong security]
		Let $\pi$ be a protocol with strategy space $\mathcal{S}$, where $s^* \in \mathcal{S}$ is the unique honest strategy profile, and let $\varepsilon>0$ be a fixed constant. We say $\pi$ enjoys \emph{$\varepsilon$-strong game-theoretic security} if the following is satisfied:
		\begin{itemize}
			\item \emph{(Completeness)} - $s^*$ is the unique subgame perfect equilibrium.
			\item \emph{(Soundness)} - There is no $s \neq s^*$ and $\varepsilon' < \varepsilon$, such that $s$ is a subgame perfect $\varepsilon'$-equilibrium.
		\end{itemize}
	\end{definition}
	
	\noindent
	\emph{A note on $\varepsilon$-soundness.}
	While our definition of $\varepsilon$-soundness may seem strange at first, it is just an artifact of the definition of subgame perfect equilibria. In order to lessen the notation burden, we want $\varepsilon$-soundness to mean that $s\neq s^*$ results in $\geq\varepsilon$ less utility than $s^*$, rather than $>\varepsilon$ less utility (which is the case when $\varepsilon' = \varepsilon$), as this results in fewer constants in our results. 
	
	\begin{lemma}
		If $\pi$ has $\varepsilon$-strong game-theoretic security then $s^*$ is an evolutionary stable strategy.
	\end{lemma}
	\begin{proof}
		We use the definition of ESS by Thomas (\cite{ess}): first, $s^*$ is clearly an equilibrium by completeness. In addition, Maynard Smith's second condition is implied by soundness.
	\end{proof}
	
	\noindent To analyze the contract from a game-theoretic perspective, we consider the contract as an extensive-form game and draw the corresponding game tree (seen in \cref{fig:gametree}). The payoff for each party is defined as their expected change in funds, where we have explicitly omitted transaction fees for simplicity. As an example, consider a dispute between a dishonest buyer and an honest seller. The buyer wins the arbitration with probability $\gamma$, earning $y$ value at no cost. They may also lose the arbitration with probability $1-\gamma$, in which case the buyer loses $x+\lambda$, for an expected payoff of $y\gamma - (x+\lambda)(1-\gamma)$. Likewise, the seller receives their payment of $x$ with probability $1-\gamma$ and loses $x+\lambda$ with probability $\gamma$, for an expected payoff of $x\,(1-\gamma) - (x+\lambda)\,\gamma-x'$. The other cases are similar and can be seen in \cref{fig:gametree}.

	\begin{lemma} \label{lemma:spe}
		There is a value of $\lambda$ such that the contract is complete if and only if the arbiter is biased in favor of honest parties.
	\end{lemma}
	\begin{proof}
		We proceed using backwards induction in the game tree. We see that the honest actions yield a strictly larger payoff if and only if the following inequalities are satisfied:
		\begin{align}
		x\,(1-\gamma) - \lambda\gamma - x' &> -x' \label{eq1}\\
		0 &> x\gamma - \lambda\,(1-\gamma) \label{eq2}\\
		y-x &> y\gamma - (x+\lambda)(1-\gamma) \label{eq3}
		\end{align}
		\cref{eq1} says that an honest seller will counter disputes from dishonest buyers. \cref{eq2} says that a dishonest seller forfeits a dispute from an honest buyer. \cref{eq3} says that a buyer will not dispute when they received $\mathit{it}$. In addition, we need $0 > -x$ and $y-x>0$ but these come from the problem statement. From \cref{eq1} we get $\lambda < x\,(\frac{1-\gamma}{\gamma})$, while \cref{eq2} yields $\lambda > x\,(\frac{\gamma}{1-\gamma})$. From \cref{eq3} we get $\lambda > x\,(\frac{\gamma}{1-\gamma}) - y$ but this is implied by \cref{eq2}. In summary, any value of $\lambda$ that achieves completeness must satisfy:
		$$x\left(\frac{\gamma}{1-\gamma}\right)<\lambda < x\left(\frac{1-\gamma}{\gamma}\right)$$
		But this is only be true when $\gamma < \frac{1}{2}$. 
	\end{proof}

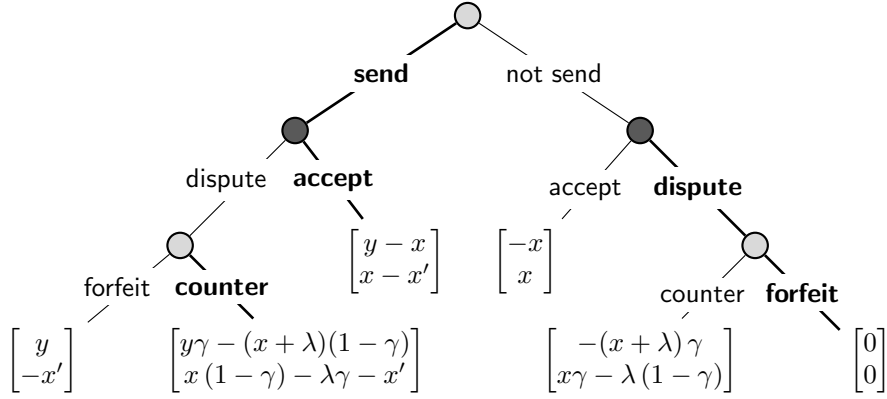
\begin{figure}
	\centering
	\resizebox{\columnwidth}{!}{%
		\begin{tikzpicture}
		\begin{scope}[every node/.style={circle,thick,draw}]
		\node[fill=gray!30!white] (A) at (3.75,0) {}; 
		\node[fill=black!30!gray] (B) at (1.5,-1.5) {}; 
		\node[fill=gray!30!white] (E) at (0,-3) {}; 
		\node[fill=black!30!gray] (H) at (6,-1.5) {}; 
		\node[fill=gray!30!white] (J) at (7.5,-3) {}; 
		\end{scope}
		\node (D) at (2.8,-3.2) {$\begin{bmatrix}y-x\\x-x'\end{bmatrix}$};
		\node (F) at (-1.8,-4.5) {$\begin{bmatrix}y\\-x'\end{bmatrix}$};
		\node (G) at (1.5,-4.5) {$\begin{bmatrix} y\gamma-(x+\lambda)(1-\gamma)\\x\,(1-\gamma) - \lambda\gamma-x'\end{bmatrix}$};
		\node (I) at (4.5,-3.2) {$\begin{bmatrix}-x\\x\end{bmatrix}$};
		\node (L) at (6,-4.5) {$\begin{bmatrix}-(x+\lambda)\,\gamma\\x\gamma-\lambda\,(1-\gamma)\end{bmatrix}$};
		\node (K) at (9,-4.5) {$\begin{bmatrix}0\\0\end{bmatrix}$};
		
		\begin{scope}[>={Stealth[black]},
		every node/.style={fill=white,rectangle}]
		\path [-] (A) edge [line width=.35mm] node {\textbf{\textsf{send}}} (B);
		\path [-] (A) edge node {\sf not send}  (H);
		\path [-] (B) edge node[xshift=-0.15cm,yshift=0.1cm] {\sf dispute}  (E);
		\path [-] (B) edge [line width=.35mm]  node {\textbf{\textsf{accept}}} (D);
		\path [-] (E) edge  node[xshift=-0.15cm,yshift=0.025cm] {\sf forfeit} (F);
		\path [-] (E) edge[line width=.35mm]  node {\textbf{\textsf{counter}}} (G);
		\path [-] (H) edge node[xshift=-0.15cm,yshift=-0.125cm] {\sf accept} (I);
		\path [-] (H) edge[line width=.35mm] node {\textbf{\textsf{dispute}}} (J);
		\path [-] (J) edge [line width=.35mm] node {\textbf{\textsf{forfeit}}} (K);
		\path [-] (J) edge node[xshift=-0.15cm,yshift=-0.1cm]{\sf counter} (L);
		\end{scope}
		\end{tikzpicture}
	}
	\caption{\small  Game tree of the smart contract after both parties have accepted the transaction. The first coordinate is the buyer payoff and the second is seller payoff. Light nodes are seller actions; dark nodes are buyer actions. The heavy edges denote the honest actions.}
	\label{fig:gametree}
\end{figure}
	
	\begin{lemma}
		Let $\varepsilon>0$ and suppose $y-\varepsilon\geq x \geq\varepsilon$. Then there is a value of $\lambda$ such that the contract is $\varepsilon$-sound if and only if $\gamma < \frac{1}{2}$ and $\varepsilon \leq x\,(1-2\gamma)$.
	\end{lemma} 
	\begin{proof}
		In any dishonest strategy profile one of the parties must choose a dishonest action. If we can show all honest actions have $\geq\varepsilon$ more utility than the dishonest actions, then the contract is $\varepsilon$-sound. But this is true when the following is satisfied:
		\begin{align}
		x\,(1-\gamma) - \lambda\gamma - x' - \varepsilon &\geq -x' \label{eq2_1}\\
		-\varepsilon &\geq x\gamma - \lambda\,(1-\gamma) \label{eq2_2}\\
		y-x - \varepsilon &\geq y\gamma - (x+\lambda)(1-\gamma) \label{eq2_3}
		\end{align}
		From \cref{eq2_1} we must have $\lambda \leq \frac{x\,(1-\gamma) - \varepsilon}{\gamma}$. From \cref{eq2_2} we get $\lambda \geq \frac{x \gamma + \varepsilon}{1-\gamma}$, and \cref{eq2_3} can be seen to be implied by \cref{eq2_2}. We also need $-\varepsilon \geq -x$ and $y-x \geq \varepsilon$ but these are given in the problem statement. In summary, any values of $\lambda$, $\varepsilon$ must satisfy:
		$$
		\frac{x \gamma + \varepsilon}{1-\gamma} \leq \lambda \leq \frac{x\,(1-\gamma) - \varepsilon}{\gamma}
		$$
		Again, we must have $\gamma < \frac{1}{2}$ since $\varepsilon>0$, while the latter condition can be established by solving for $\varepsilon$. 
	\end{proof}
	
	\begin{theorem}
		The contract has $x\,(1-2\gamma)$-strong game-theoretic security whenever $\gamma < \frac{1}{2}$ and $\lambda = x$.
	\end{theorem}
	\begin{proof}
		Since $\gamma< \frac{1}{2}$ the conditions for completeness are satisfied for $\lambda=x$. For soundness, let $\varepsilon = x\,(1-2\gamma) $. We choose a value of $\lambda$ that satisfies the lower bound:
		\begin{align*}
		\lambda \geq \frac{x\gamma + \varepsilon}{1-\gamma}
		= \frac{x\gamma + x\,(1-2\gamma)}{1-\gamma} = x & \qedhere
		\end{align*}
	\end{proof}
	
	%\noindent We now consider repaying the winning party some of the loser's wager, i.e. the winner is repaid $x+(1+\alpha)\lambda$ for some constant $0 \leq \alpha \leq 1$. This will allow us to obtain $\varepsilon$-strong security for an arbitrary constant 
	
	\section{Wager functions}
	In this section we consider a generalization of the previous protocol that allows us to obtain various tradeoffs between security and wager size. Each parties submits a wager such that if a party loses the arbitration, they will lose $\ell$, while a winner is paid $\omega$, for some functions $\omega,\ell$. We will assume that $-\omega > \ell$ such that winning is preferred over losing. The upper bound of $\gamma < \frac{1}{2}$ seems inherent, and we show that it cannot be beat, even by arbitrary $\omega,\ell$.
	
	\begin{lemma}
		For arbitrary rebate functions $\omega,\ell$ the contract is only complete if the arbiter is biased in favor of honest parties.
	\end{lemma}
	\begin{proof}
		Consider a seller who has to decide whether or not to counter a dispute from the buyer. Regardless of whether the seller is honest or not, they have to decide whether to forfeit or counter. If the seller is honest we want them to counter the dispute, i.e. $\omega\,(1- \gamma) - \ell\gamma > 0$. If the seller is dishonest we want them to forfeit, i.e. $\omega\gamma - \ell\,(1-\gamma) < 0$. That is,
		$$
		\omega\gamma - \ell\,(1-\gamma) < \omega\,(1- \gamma) - \ell\gamma.
		$$
		Since we have $-\omega > \ell$ this can only be true for $\gamma < \frac{1}{2}$.
	\end{proof}
	
	\subsection{Winner rebates}
	One natural choice of wager function is to pay the winner back a rebate. The winner, in addition to winning back their wager, also wins the loser's wager. It turns out this allows us to get $\varepsilon$-strong security for arbitrary choice of $\varepsilon>0$.
	
	\begin{theorem}
		With a winner rebate of size $\lambda$, the contract has $\varepsilon$-strong security whenever $\gamma < \frac{1}{2}$ and $\lambda = \frac{x\gamma + \varepsilon}{1-2\gamma}$.
	\end{theorem}
	\begin{proof}
		We proceed using backwards induction. Again, we only need to consider the case of a seller facing a dispute as this implies the other cases. In particular, $\varepsilon$-strong security is achieved when the following inequalities are satisfied:
		\begin{align}
		(x+\lambda)\,(1-\gamma) - \lambda \gamma -\varepsilon &\geq 0 \label{eq4_1}\\
		-\varepsilon &\geq (x+\lambda)\,\gamma - \lambda\,(1-\gamma) \label{eq4_2}
		\end{align}
		From either of these equations we get $\lambda \geq \frac{x\gamma + \varepsilon}{1-2\gamma}$ for any $\varepsilon>0$, showing the claim.
	\end{proof}
	\noindent The downside to this is that it results in larger wagers: suppose we let $\varepsilon = x\,(1-2\gamma)$, the maximum value in the old contract, then the new wager is:
	$$\lambda = \frac{x\gamma + x\,(1-2\gamma)}{1-2\gamma} = x + \frac{x\gamma}{1-2\gamma} > x$$
	which is always larger than the old wager. This is natural in a sense: since we expect to win back the wager by disputing, the wager needs to be larger to offset the increased incentive to issue a false dispute.
	
	\subsection{Withholding wagers}
	We now consider what happens the wager is withheld even for the winning party. This will allow us to get $\lambda = \frac{1}{2}x$, however the lower bound of $\lambda=\Omega(x)$ seems inherent for this construction.
	
	\begin{theorem}
		When the contract withholds all wagers, it has $\frac{1}{2}x\,(1-2\gamma)$-strong game-theoretic security when $\gamma<\frac{1}{2}$, and $\lambda=\frac{1}{2}x$.
	\end{theorem}
	\begin{proof}
		Follows immediately using backwards induction.
	\end{proof}

	\section{Coin toss arbitration}
	
	In this section we consider the special case in which the output of the arbiter is independent of the evidence being submitted, i.e. $\gamma=\frac{1}{2}$. The advantage of this is that we can implement such an arbiter using a cryptographic protocol. However, we showed that strong game-theoretic security is only possible when $\gamma<\frac{1}{2}$ so we need to relax our security definition.
	
	\begin{definition}[Weak security]
		Let $\pi$ be a protocol with strategy space $\mathcal{S}$ where $s^* \in \mathcal{S}$ is the unique honest strategy profile. We say $\pi$ enjoys \emph{weak game-theoretic security} if $s^*$ is a subgame perfect equilibrium.
	\end{definition}
	
	\noindent While this guarantees that being honest is an equilibrium strategy it does not provide a strict incentive to do so. 
	
	\begin{theorem}
		Using a coin toss arbiter, the contract has weak game-theoretic security for $\gamma=\frac{1}{2}$ and $\lambda=x$.
	\end{theorem}
	\begin{proof}
		For $s^*$ to be a subgame perfect equilibrium, there must be no $s \neq s^*$ that achieves a strictly larger payoff. As before, this is achieved when the following inequalities are satisfied:
		\begin{align}
		x\,(1-\gamma) - \lambda\gamma - x' &\geq -x' \label{eq3_1}\\
		0 &\geq x\gamma - \lambda\,(1-\gamma)  \label{eq3_2}\\
		y-x &\geq y\gamma - (x+\lambda)(1-\gamma) \label{eq3_3}
		\end{align}
		Again, \cref{eq3_3} is implied by \cref{eq3_2}. For the remaining inequalities we need,
		$$
		x\,(1-\gamma) - \lambda\gamma \geq x\gamma - \lambda\,(1-\gamma)
		$$
		which solves to $\lambda=x$ for $\gamma=\frac{1}{2}$.
	\end{proof}

	\subsection{Construction}
	We can implement the coin toss arbiter using a variant of Blum's coin flipping protocol \cite{blum_cointoss,ivan_commitments}. Suppose we have a commitment scheme, and let \textsf{commit} be the commitment function. Then the arbitration proceeds as follows:
	\begin{enumerate}
		\item $S$ samples a random bit $b_S \in_R \{0,1\}$, and a random string $r \in_R \{0,1\}^\kappa$. 
		
		\item $S$ computes $C \gets \textsf{commit}(b,r)$ and submits $C$ to the smart contract.
		
		\item $B$ samples a random bit $b_B \in_R \{0,1\}$ and submits it to the smart contract.
		
		\item $S$ submits $b_S$ and $r$ to the blockchain.
		
		\item The smart contract verifies that $b_S, r$ is a valid opening of $C$: if not, let $b := 0$. Otherwise let $b := b_S \oplus b_B$. 
		
		\item The smart contract transfers $x+\lambda$ to $S$ if and only if $b=1$, and transfers $x+\lambda$ to $B$ otherwise.
	\end{enumerate}
	
	\noindent If at some point either party times out, it is assumed they forfeited, and the funds held in escrow are released to the other party. Analysis of the protocol is straight forward: the output is clearly uniform, and security reduces to that of the commitment scheme. 
	
	From a cryptographic perspective, this protocol is not satisfactory as it does not satisfy \emph{fairness}: $S$ can choose not to complete step 4 and simply abort the protocol without revealing the output to $B$. However, this is not an issue for our application, since $S$ loses the dispute by doing so.

	\section{Jury based arbitration}
	
	In this section we consider alternatives to the coin toss arbiter based on jury systems, with the hope of obtaining $\gamma < \frac{1}{2}$. The naive solution is to simply sample a random jury, give them the evidence, and decide the outcome by majority decision. This is essentially how jury duty works. Unlike jury duty however, here jurors have no incentive to actually spend time assessing the evidence: the optimal strategy for a juror is to simply flip a coin to decide the outcome, essentially reducing the arbiter to the coin toss protocol.
	
	However, there are promising candidates for jury based arbiters, an example of which is Kleros \cite{kleros}. Kleros is a decentralized arbiter built on top of Ethereum. It is essentially a generalization of the jury system described above with a number of additional features, such as a built-in appeal system, and arborescence in the form of subcourts. Here, jurors can opt-in on a case-by-case basis, and are incentivized to answer honestly by the use of penalties/rewards: a juror who votes in accordance with the majority decision is rewarded, and otherwise they are penalized. This is justified by the use of a \emph{focal point} (or Schelling point), defined as the strategy people choose by default in absence of communication. As an example, consider the coordination game shown in \cref{fig:coordination-game}: two parties are shown four squares, one of which is distinguished. If the parties choose the same square they are rewarded, otherwise they receive nothing. Technically speaking, any individual square is an equilibrium strategy. In practice however, most people would expect the other party to choose the distinguished square; hence they will do so themselves. Thus choosing the distinguished square is a focal point. The underlying assumption of Kleros is that other jurors \emph{expect} others to vote honestly, and hence do so as well. Unfortunately, there is no empirical study of the error rate of Kleros, so whether the focal point of Kleros is actually ``truth'' remains conjecture at this point.

	\begin{figure}
		\centering
		\begin{tikzpicture}
		\draw[draw=black] (0,0) rectangle ++(0.6,0.6);
		\draw[draw=black] (1,0) rectangle ++(0.6,0.6);
		\draw[fill=black] (2,0) rectangle ++(0.6,0.6);
		\draw[draw=black] (3,0) rectangle ++(0.6,0.6);
		\end{tikzpicture}
		\caption{\small A coordination game where two parties are tasked with choosing the same square. While choosing any square is technically an equilibrium strategy, the focal point is choosing the black square, as most people will default to choosing this in absence of communication.}
		\label{fig:coordination-game}
	\end{figure}
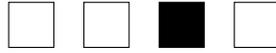
	
	\section{Choice of blockchain}
	We did not consider any specific blockchain in the previous sections: our contract works for any blockchain with the capability to execute smart contracts. As a result, the contract inherits many properties of the underlying blockchain, which means the contract can be instantiated in a variety of ways. In this section we consider various instantiations of the contract in different blockchains.
	
	\paragraph{Pseudonymity} Instantiating the contract on a public ledger such as Ethereum is the most straight-forward solution. Of course, this means that all transactions are public and available to other parties. However, for some applications this can be considered a feature: having access to the transaction history of a seller gives an indication of how likely they are to cheat and holds the parties somewhat responsible for their actions. In fact, Dellarocas \cite{reputation_systems} shows that under the right conditions, a long-lived seller has strong incentive to behave honestly when faced with many short-term buyers. The incentive is strongest in the initial phase where the seller has to work hard to build up a good reputation, and diminishes as their reputation increases.
	
	\paragraph{Full anonymity} It is possible to use the contract to facilitate fully anonymous trading by using a blockchain with built-in anonymity (Monero, Zcash, etc.). Doing so necessitates the use of the coin toss arbiter, as it is impossible to remain anonymous when submitting an evidence string containing personal information. However, this makes it impossible to enforce regulation on the goods being transacted, and it is likely such a market would be used for criminal activity.
	
	\paragraph{Revocable anonymity} We can also make the contract comply with all laws and regulations by using a blockchain that supports revocable anonymity \cite{revocable_anonymity}. To register in such a blockchain you need to identify yourself using your passport or similar document. Parties can then create new accounts that are anonymous. It can then be proven in zero knowledge that a given buyer is eligible to purchase the goods in question, for instance by enforcing age restrictions, or ensuring certain goods can only be purchased by those with suitable license, such as hunting rifles. In addition, the accounts can be \emph{deanonymized} under suitable conditions, say if illegal behavior is suspected. This requires an agreement between several organizations, e.g the police, the courts and possibly other organizations and should serve as an incentive not to engage in criminal activity.
	
	\section{Transaction fees}
	
	Our analysis assumes transaction fees are negligible, which is not the case in practice. In this section we consider adding transaction fees to our model. Doing so in general is tricky business and is very specific to the implementation and blockchain of choice. Instead, we adopt a simplified approach where playing a move in the game tree has a unit cost of $\tau$ for some $\tau > 0$, the only exception being the default action in case of timeouts: a player can always time out to choose the default action at zero cost.
	
	\begin{lemma}
		With transaction fees of size $\tau$, the contract is complete if and only if the arbiter is biased in favor of honest parties, the transaction fee is bounded $\tau < x\,(1-\gamma)-\lambda\gamma$, and the item is of sufficient value, $x-x' > \tau$.
	\end{lemma}
	\begin{proof}
		We proceed using backwards induction in the game tree. It is not hard to see we still need $\gamma<\frac{1}{2}$. Consider a seller faced with a dispute. When he is dishonest his incentive to be honest is increased by $\tau$, while the converse is true when he is honest. This yields the following inequality:
		\begin{equation}
		x\,(1-\gamma)-\lambda\gamma-x'-\tau > -x'
		\end{equation}
		Which solves to $\tau < x\,(1-\gamma)-\lambda\gamma$. Now consider a buyer. If she did receive the item, her incentive to accept has only increased by $\tau$. If she did not receive the item, her added cost of $\tau$ for issuing a dispute must outweight the size of the payment. This means we must have $x>\tau$. Finally, consider a seller deciding whether to send or not. If he does not send he incurs a cost of 0, while accepting gives $x-x'-\tau>0$.
	\end{proof}
	
	\begin{lemma}
		With transaction fees of size $\tau$, the contract is $\varepsilon$-sound only when $\varepsilon \leq x\,(1-2\gamma)-\tau$, and the transaction fee is bounded $\tau < x\,(1-2\gamma)$.
	\end{lemma}
	\begin{proof}
		We proceed using backwards induction. Similar arguments as before gives us $\gamma<\frac{1}{2}$ and $\varepsilon \leq x\,(1-2\gamma)$. The bound on $\tau$ is established by considering a seller faced with a dispute. If he is honest we want him to counter which gives:
		\begin{equation}
		x\,(1-\gamma) - \lambda\gamma -x' - \tau > \varepsilon-x'
		\end{equation} 
		This solves to $\lambda < \frac{x\,(1-\gamma)-\varepsilon-\tau}{\gamma}$. Likewise, if he is dishonest we want him to forfeit which gives:
		\begin{equation}
		x\gamma - \lambda\,(1-\gamma) - \tau  < - \varepsilon
		\end{equation} 
		This gives $\lambda > \frac{x\gamma+\varepsilon+\tau}{1-\gamma}$. Combining these, we get:
		$$\frac{x\,(1-\gamma)-\varepsilon-\tau}{\gamma} > \frac{x\gamma+\varepsilon+\tau}{1-\gamma}$$
		which solves to $\tau < x\,(1-2\gamma)-\varepsilon$.
	\end{proof}
	
	\noindent With these two lemmas in place we can prove game-theoretic security using similar arguments as before. This allows us to establish the following:
	
	\begin{theorem}
		With transaction fees of size $\tau$, the contract has $[x\,(1-2\gamma)-\tau]$-strong game-theoretic security when $\lambda = x$ and $x-x'>\tau$.
	\end{theorem}
	
	\subsection{Liveness}
	
	The protocol as presented thus far does not ensure liveness. We remedied this by introducing timeouts, such that players default to choosing an action when they did not respond. For some applications this may not be satisfactory, in that we want the parties to respond as quickly as possible rather than waiting for the timeout. We will fix this using an idea of Asgaonkar and Krishnamachari \cite{buyer_seller_dilemma}. We require both parties to submit a deposit in order to accept the contract. If a party times out, it loses its deposit. This incentivizes the parties to answer within the timeout. However, playing a move at any time before the timeout is still an equilibrium. We can fix the problem by withholding an amount of the deposit proportional to time taken. However, this unfairly punishes honest parties which may incur a short delay for legitimate reasons. Instead, the protocol is parameterized by two timeouts $T_\text{threshold} < T_\text{timeout}$. $T_\text{threshold}$ is a ``reasonable'' time before which parties are not punished, and $T_\text{timeout}$ is the timeout where a default action is taken. Let $D$ be the size of the deposit. A party makes a decision at time $t$ is paid $\rho(t)$ defined by:
	$$\rho(t) = \begin{cases}D & \text{if $t \leq T_\text{threshold}$}\\D\left(1-\frac{t - T_\text{threshold}}{ T_\text{timeout} - T_\text{threshold}}\right) & \text{if $T_\text{threshold} < t < T_\text{timeout}$} \\ 0 & \text{o.w.}\end{cases}$$
	This ensures that parties have largest incentive to answer within reasonable time, and otherwise to answer as quickly as possible.
	
	\section{Multiparty transactions}
	In this section we consider generalizing the contract to more parties. Consider a \emph{multiparty transaction} with $n$ parties: each party may possibly transact with every other party (acting both as seller and buyer). Note that we could just invoke the two-party contract for each pair, however this would incur transaction fees $O(n^2)$ times which we would like to avoid. We present a contract that enables a set of $n$ parties to make any number pairwise transactions using $O(n)$ transaction fees, which is easily seen to be optimal.
	
	We denote by $\P_1, \ldots \P_j$ the $n$ players, and let $x_{ij}$ be the size of the payment from $\P_i$ to $\P_j$. We will essentially run the two-party protocol $n^2$ times in parallel, and compress the bookkeeping using a procedure similar to how you would split the bill at a restaurant, the algorithm is as follows:
	
	\begin{enumerate}
		\item $\P_i$ deposits $\sum_{j=1}^n x_{ij}$ to the contract.
		\item All items are delivered off-chain.
		\item $\P_i$ produces a \emph{dispute vector} $d_i$ where $d_{ij}=0$ if $\P_i$ does not dispute the item received from $\P_j$, and $d_{ij}=1$ otherwise.
		\item $\P_i$ deposits $\sum_{j=1}^n d_{ij} \lambda_{ij}$ to the contract, where $\lambda_{ij} = x_{ji}$. 
		\item $\P_{i}$ is given $\{d_{ji}\}_{j=1}^n$, and produces a \emph{counter vector} $c_i$ where $c_{ij}=1$ if $\P_i$ counters the dispute issued by $\P_j$, and $c_{ij}=0$ otherwise. 
		\item $\P_i$ deposits $\sum_{j=1}^n c_{ij} \lambda_{ji}$ to the contract.
		\item The contract samples a matrix $b \in_R \{0,1\}^{n\times n}$ uniformly at random, and sends $\sum_{j=1}^n(x_{ji} - b_{ij} c_{ij} d_{ij} \lambda_{ij})$ to $\P_i$.
	\end{enumerate}
	
	\noindent It is not hard to verify that each player has the correct balance afterwards. Security follows from the \emph{one-shot deviation principle}: all games are essentially independent so we may suppose they are played in some order. We now proceed using backwards induction: game-theoretic security implies the last game must be secure since there is no venue for punishing misbehavior. But then the same must be true for the second-to-last game, and so on. 
	
	Finally, the number of transaction fees is $O(n)$ since each player makes three deposits in the worst-case, and receives one withdrawal. It is not hard to see this is optimal since each player needs to participate in the contract. 
	
	\section{Conclusion}
	In this paper we proposed a smart contract for trading any physical good or service using a smart contract as escrow. The contract settles disputes by a wager between the buyer and seller, where the parties wager that they can convince an arbiter of their honesty. The contract was shown to be secure in a strong game-theoretic sense for certain values of the size of the wager, assuming the arbiter is biased in favor of honest parties. The contract can be shown secure in a weaker sense when the arbiter is replaced by a random coin toss. Finally, we showed how to generalize the contract to multiple parties in a way that amortizes the transaction fees.
	
	\paragraph{Future work} In order for the contract to be used in practice, there is need for an implementation in e.g. Solidity and formally prove its correctness.
	
	\bibliographystyle{plain}
	\bibliography{paper}
\end{document}